\documentclass[11pt]{article}      %STOC requirement http://acm-stoc.org/stoc2023/cfp.html
\usepackage[T1]{fontenc}
\usepackage[margin=1in]{geometry}  %STOC requirement http://acm-stoc.org/stoc2023/cfp.html
\usepackage{libertine} % font
\usepackage{inconsolata} % texttt font
\usepackage{amsthm}
\usepackage{amsmath}
\usepackage{amssymb}
\usepackage{mathrsfs}
\usepackage{thm-restate}
\usepackage[unicode,colorlinks=true,
            linkcolor=blue,
            citecolor=blue,
            urlcolor=blue]{hyperref}
\usepackage[capitalize,nameinlink]{cleveref}
\usepackage{url}
\usepackage[size=normalsize]{todonotes}
\usepackage{enumerate}
\usepackage[all,defaultlines=3]{nowidow}
\usepackage[mathlines]{lineno}
\usepackage{xcolor}
\usepackage{multicol}
\usepackage{tabularx}
\usepackage{colortbl}
\usepackage{float}
\usepackage{caption}
\usepackage{tikz}

\usepackage{enumitem}
\setlist[itemize]{topsep=4pt,itemsep=3pt,parsep=0pt} 
\setlist[enumerate]{topsep=4pt,itemsep=3pt,parsep=0pt}
\setlist[description]{topsep=4pt,itemsep=3pt,parsep=0pt}

\newtheorem{theorem}{Theorem}
\newtheorem{corollary}{Corollary}

\newtheorem{lemma}{Lemma}
\newtheorem{example}{Example}
\newtheorem{claim}{Claim}

\theoremstyle{definition}

\theoremstyle{plain}

\newenvironment{claimproof}[1][\proofname]{%
  \begin{proof}[#1]%
}{%
  \end{proof}%
}

\def\phi{\varphi}
\def\epsilon{\varepsilon}
\renewcommand{\emptyset}{\varnothing}

\newcommand{\N}{\mathbb{N}}
\def\CC{\mathscr{C}}
\newcommand{\Cc}{\mathscr{C}}
\newcommand{\Oof}{\mathcal{O}}
\newcommand{\PP}{\mathcal P}
\renewcommand{\SS}{\mathcal{S}}
\newcommand{\Gyarfas}{\textnormal{Gyárfás}}
\newcommand{\GyarfasSub}{\textnormal{GyárfásSub}}

\newcommand{\funding}{NM received funding from the European Research Council (ERC) with grant agreement No.\ 101126229 -- {\sc buka}.
}

\begin{document}
\title{The Parameterized Complexity of Independent Set and More\\ when Excluding a Half-Graph, Co-Matching, or Matching\thanks{\funding}}
\date{}
\author{
  Jan Dreier
  \\
  \small{TU Wien}
  \\
  \small{\texttt{dreier@ac.tuwien.ac.at}}
  \and
  Nikolas M\"ahlmann
  \\
  \small{University of Warsaw}
  \\
  \small{\texttt{maehlmann@mimuw.edu.pl}}
  \and
  Sebastian Siebertz
  \\
  \small{University of Bremen}
  \\
  \small{\texttt{siebertz@uni-bremen.de}}
}

\maketitle

\begin{abstract}
A theorem of Ding, Oporowski, Oxley, and Vertigan implies that any sufficiently large twin-free graph contains a large matching, a co-matching, or a half-graph as a semi-induced subgraph.
The sizes of these unavoidable patterns are measured by the \emph{matching index}, \emph{co-matching index}, and \emph{half-graph index} of a graph.
Consequently, graph classes can be organized into the eight classes determined by which of the three indices are bounded.

We completely classify the parameterized complexity of \textsc{Independent Set}, \textsc{Clique}, and \textsc{Dominating Set} across all eight of these classes.
For this purpose, we first derive multiple tractability and hardness results from the existing literature, and then proceed to fill the identified gaps. 
Among our novel results, 
we show that \textsc{Independent Set} is fixed-parameter
tractable on every graph class where the half-graph and co-matching indices are simultaneously bounded.
Conversely, we construct a graph class with bounded half-graph index (but unbounded co-matching index), for which the problem is W[1]-hard.

For the W[1]-hard cases of our classification, we review the state of approximation algorithms. Here, we contribute an approximation algorithm for \textsc{Independent Set} on classes of bounded half-graph index.
\end{abstract}

\begin{picture}(0,0)
\put(422,-240)
{\hbox{\includegraphics[width=40px]{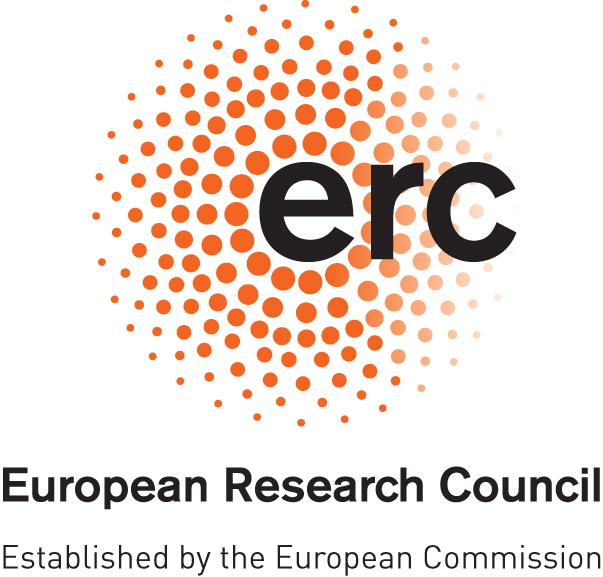}}}
\put(412,-300)
{\hbox{\includegraphics[width=60px]{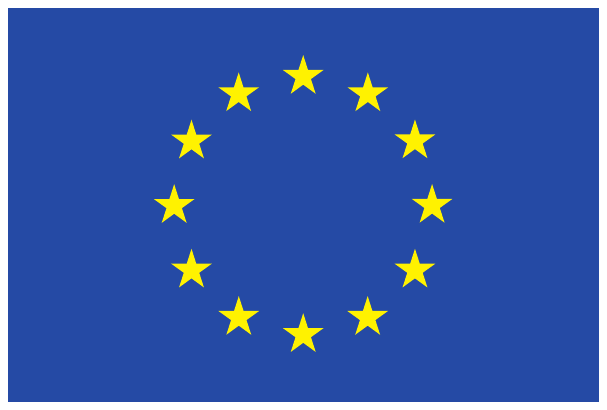}}}
\end{picture}

\newpage

\section{Introduction}

\textsc{Independent Set}, \textsc{Clique}, and \textsc{Dominating Set} are three of the most fundamental NP-hard graph problems.
From the parameterized perspective, all three are intractable on general graphs as well:
\textsc{Independent Set} and \textsc{Clique} are $\mathrm{W}[1]$-complete~\cite{downeyfellows1995w1},
and \textsc{Dominating Set} is $\mathrm{W}[2]$-complete~\cite{downeyfellows1995basic} when parameterized by the solution size.
This motivates identifying structural restrictions under which these problems become tractable.

A natural approach is to forbid certain bipartite patterns in the adjacency relation.
Given two disjoint vertex sets $A = \{a_1, \ldots, a_t\}$ and $B = \{b_1, \ldots, b_t\}$, we can ask: what are the canonical ways to organize the edges between them? 
The three patterns we will be interested in are the following:
\begin{itemize}
    \item \makebox[4cm][l]{the \emph{matching} $M_t$:} \makebox[4cm][l]{$a_i b_j \in E \Leftrightarrow i = j$} (equality)
    \hfill
    \begin{tikzpicture}[baseline=-0.5ex, scale=0.35]
        \foreach \i in {1,...,4} {
            \fill ({\i*1.2},1) circle (3pt);
            \fill ({\i*1.2},-1) circle (3pt);
            \draw ({\i*1.2},1) -- ({\i*1.2},-1);
        }
    \end{tikzpicture}
    \item \makebox[4cm][l]{the \emph{co-matching} $C_t$:} \makebox[4cm][l]{$a_i b_j \in E \Leftrightarrow i \neq j$} (inequality)
    \hfill
    \begin{tikzpicture}[baseline=-0.5ex, scale=0.35]
        \foreach \i in {1,...,4} {
            \fill ({\i*1.2},1) circle (3pt);
            \fill ({\i*1.2},-1) circle (3pt);
        }
        \foreach \i in {1,...,4} {
            \foreach \j in {1,...,4} {
                \ifnum\i=\j\else
                    \draw ({\i*1.2},1) -- ({\j*1.2},-1);
                \fi
            }
        }
    \end{tikzpicture}
    \item \makebox[4cm][l]{the \emph{half-graph} $H_t$:} \makebox[4cm][l]{$a_i b_j \in E \Leftrightarrow i \leq j$} (order)
    \hfill
    \begin{tikzpicture}[baseline=-0.5ex, scale=0.35]
        \foreach \i in {1,...,4} {
            \fill ({\i*1.2},1) circle (3pt);
            \fill ({\i*1.2},-1) circle (3pt);
        }
        \foreach \i in {1,...,4} {
            \foreach \j in {\i,...,4} {
                \draw ({\i*1.2},1) -- ({\j*1.2},-1);
            }
        }
    \end{tikzpicture}
\end{itemize}

In some sense, these three patterns form the fundamental building blocks of complexity, as made precise in the following theorem proven independently by Ding et al.~\cite[Cor.\ 2.4]{DBLP:journals/jct/DingOOV96}, Alekseev~\mbox{\cite[Thm.\ 3]{alekseev1997lower}}, and Gravier et al.~\cite[Thm.\ 2]{gravier2004}. (See also \cite{HONS2026104303} for polynomial bounds in graphs of bounded VC-dimension.)

\begin{theorem}[{\cite{DBLP:journals/jct/DingOOV96,alekseev1997lower,gravier2004}}]\label{thm:trichotomy}
    There exists a function $Q : \N \to \N$ such that every bipartite graph without twins and with at least $Q(h)$ vertices on one side contains a matching, co-matching, or half-graph of order $h$ as an induced subgraph.
\end{theorem}

With a bit of work, the bipartite statement extends to general graphs. 
Two vertices $u,v$ in a graph $G$ are \emph{twins} if $N(u)\setminus\{v\}=N(v)\setminus\{u\}$. 
This is an equivalence relation~(see, e.g., \cite[Lemma 1]{lampis2012algorithmic}), and we call its classes the \emph{twin classes}.
The \emph{neighborhood diversity} of $G$, introduced by Lampis~\cite{lampis2012algorithmic}, is the number of twin classes of \(G\).
Each class induces either a clique or an independent set, and between any two classes the adjacency is either complete or empty.
Thus, graphs of small neighborhood diversity are extremely simple.
The interesting case is when neighborhood diversity is large: then one of the three patterns must appear in $G$ in a semi-induced way.
Formally, a bipartite graph $H$ appears \emph{semi-induced} in $G$ if there exist $A,B\subseteq V(G)$ such that the bipartite
subgraph between $A$ and $B$ is isomorphic to $H$. The \emph{matching index},
\emph{co-matching index}, and \emph{half-graph index} of $G$ are the maximum orders
of semi-induced matchings, co-matchings, and half-graphs, respectively. For general
graphs, \Cref{thm:trichotomy} therefore implies the following corollary (see \Cref{sec:proof-nd-trichotomy}).

\begin{restatable}{corollary}{NdTrichotomy}\label{cor:nd-trichotomy}
    There exists a function \(Q : \N \to \N\) such that every graph with neighborhood diversity at least~\(Q(h)\)
    has either matching index, co-matching index, or half-graph index at least \(h\).
\end{restatable}

% This corollary explains why semi-induced variants of the three patterns are the ``right'' ones to study.
A graph class has \emph{bounded} matching/co-matching/half-graph index if the corresponding parameter is bounded by a constant across all graphs in the class. We also call these classes \emph{matching-, co-matching-,} and \emph{half-graph-free}.
This yields a natural classification of graph classes into eight cases, depending on which combination of the three indices is bounded.

\subsection*{Contribution 1: A complete classification}

Our main contribution is a complete classification of the parameterized complexity of \textsc{Independent Set}, \textsc{Clique}, and \textsc{Dominating Set} across all eight combinations of boundedness for the matching, co-matching, and half-graph index.
The results are summarized in \Cref{tab:classification}.
Observe that the three indices enjoy the following nice duality properties.
Fix a graph class $\CC$ and let $\overline{\CC}$ be the class containing the complements of the graphs in $\CC$.
\begin{itemize}
    \item $\CC$ is matching-free if and only if $\overline{\CC}$ is co-matching-free.
    \item $\CC$ is half-graph-free if and only if $\overline{\CC}$ is half-graph-free.
\end{itemize}
As cliques are independent sets in the complement graph,
a complete classification of the \textsc{Independent Set} column also yields a classification of the \textsc{Clique} column.
Hence, we focus our exposition on \textsc{Independent Set} and \textsc{Dominating Set}.

\newcommand{\yesmark}{\ensuremath{\checkmark}}
\newcommand{\nomark}{\ensuremath{\times}}

% Define toned-down colors
\definecolor{ptimegreen}{RGB}{200, 230, 200}
\definecolor{fptgreen}{RGB}{220, 240, 220}
\definecolor{hardred}{RGB}{245, 215, 215}
\definecolor{indexgray}{gray}{0.88}

% Result cells with colors
\newcommand{\rcellptime}[2]{\cellcolor{ptimegreen}#1\hfill{\scriptsize #2}}
\newcommand{\rcellfpt}[2]{\cellcolor{fptgreen}#1\hfill{\scriptsize #2}}
\newcommand{\rcellhard}[2]{\cellcolor{hardred}#1\hfill{\scriptsize #2}}

\begin{table}[h]
\centering
\renewcommand{\arraystretch}{1.1}
\begin{tabularx}{\textwidth}{|>{\cellcolor{indexgray}}c|>{\cellcolor{indexgray}}c|>{\cellcolor{indexgray}}c|X|X|X|}
\hline
\rowcolor{indexgray}
$m$ & $c$ & $h$ & \centering \textsc{Independent Set} & \centering \textsc{Clique} & \centering\arraybackslash \textsc{Dominating Set} \\
\hline
\hline
\yesmark & \yesmark & \yesmark
    & \rcellptime{uniform PTime}{\cite{lampis2012algorithmic}}
    & \rcellptime{uniform PTime}{\cite{lampis2012algorithmic}}
    & \rcellptime{uniform PTime}{\cite{lampis2012algorithmic}} \\
\hline
\hline
\yesmark & \yesmark & \nomark
    & \rcellptime{PTime}{\cite{BUIXUAN201366}}
    & \rcellptime{PTime}{\cite{BUIXUAN201366}}
    & \rcellptime{PTime}{\cite{BUIXUAN201366}} \\
\hline
\yesmark & \nomark & \yesmark
    & \rcellptime{PTime}{\cite{BUIXUAN201366}}
    & \rcellfpt{uniform FPT}{Thm.~\ref{thm:IS-kernel}}
    & \rcellptime{PTime}{\cite{BUIXUAN201366}} \\
\hline
\nomark & \yesmark & \yesmark
    & \rcellfpt{uniform FPT}{Thm.~\ref{thm:IS-kernel}}
    & \rcellptime{PTime}{\cite{BUIXUAN201366}}
    & \rcellfpt{uniform FPT}{\cite{FabianskiPST19}} \\
\hline
\hline
\nomark & \nomark & \yesmark
    & \rcellhard{W[1]-hard}{Thm.~\ref{thm:hardness-is-halfgraph-free}}
    & \rcellhard{W[1]-hard}{Thm.~\ref{thm:hardness-is-halfgraph-free}}
    & \rcellhard{W[1]-hard}{\cite{cygan2015parameterized}, Thm.~\ref{thm:hardness-ds-halfgraph-free}} \\
\hline
\nomark & \yesmark & \nomark
    & \rcellhard{W[1]-hard}{\cite{marx05}, Cor.~\ref{cor:unitquare_comatching}}
    & \rcellptime{PTime}{\cite{BUIXUAN201366}}
    & \rcellhard{W[1]-hard}{\cite{marx06}, Cor.~\ref{cor:unitquare_comatching}} \\
\hline
\yesmark & \nomark & \nomark
    & \rcellptime{PTime}{\cite{BUIXUAN201366}}
    & \rcellhard{W[1]-hard}{\cite{marx05}, Cor.~\ref{cor:unitquare_comatching}}
    & \rcellptime{PTime}{\cite{BUIXUAN201366}} \\
\hline
\hline
\nomark & \nomark & \nomark
    & \rcellhard{W[1]-hard}{\cite{cygan2015parameterized}}
    & \rcellhard{W[1]-hard}{\cite{cygan2015parameterized}}
    & \rcellhard{W[2]-hard}{\cite{cygan2015parameterized}} \\
\hline
\end{tabularx}
\caption{
    Classification of \textsc{Independent Set}, \textsc{Clique}, and \textsc{Dominating Set}.
    Each row considers graph classes \(\Cc\)
    that are matching-free ($m$), co-matching-free ($c$), or half-graph-free ($h$), as indicated by \nomark~and~\yesmark.
    We write \emph{W[1]-hard} if there exists a graph class \(\Cc\) with the considered restrictions
    on which the problem is W[1]-hard.
    We write \emph{FPT} and \emph{PTime} to indicate that the problem can be solved in time \(f(k) \cdot \textnormal{poly}(n)\) for some function~\(f\) and solution size \(k\),
    or \(\textnormal{poly}(n)\) for each graph class \(\Cc\) with the considered restrictions.
    Note that the degree of \(\textnormal{poly}(n)\) may depend on \(\Cc\).
    In contrast, we write \emph{uniform FPT} or \emph{uniform PTime} if the degree of \(\textnormal{poly}(n)\) (but not necessarily the function~\(f\)) can be bounded uniformly, independent of the considered graph class.
}\label{tab:classification}
\end{table}

\paragraph{Known entries.}
%We briefly explain the table.
When all three indices are bounded, neighborhood diversity is bounded, and all three problems are polynomial-time solvable~\cite{lampis2012algorithmic}.
More generally, when the matching index is bounded, every branch decomposition has bounded \emph{mim-width} (maximum induced matching width), yielding polynomial-time algorithms for \textsc{Independent Set} and \textsc{Dominating Set}~\cite{BUIXUAN201366}.
These results account for the PTime entries in \Cref{tab:classification}.

For parameterized complexity, Fabiański et al.~\cite{FabianskiPST19} show that \textsc{Dominating Set} is FPT on classes that are simultaneously half-graph-free and co-matching-free.
On general graphs, \textsc{Independent Set} and \textsc{Clique} are W[1]-hard and \textsc{Dominating Set} is W[2]-hard~\cite{cygan2015parameterized}.

\paragraph{New entries.}
Several hardness entries use previously known reductions, but the connection to the pattern indices had not been made explicit.
For instance, it is known that \textsc{Independent Set}~\cite{marx05} and \textsc{Dominating Set}~\cite{marx06} are W[1]-hard on unit square graphs; we observe that these graphs have bounded co-matching index (\Cref{cor:unitquare_comatching}).
In such cases, we cite both the original hardness result and our observation connecting it to the pattern indices.

Our main new tractability result is that \textsc{Independent Set} is FPT on classes that are simultaneously half-graph-free and co-matching-free (\Cref{thm:IS-kernel}).
We obtain this via indiscernibility-based methods inspired by~\cite{DreierMST23}:
in such graphs, long vertex sequences arise where outside vertices connect to almost all or almost none of the sequence elements, enabling efficient reduction rules.
Note that \textsc{Independent Set} and \textsc{Dominating Set} are NP-hard on planar graphs~\cite{gareyjohnson1979}, which are co-matching-free and half-graph-free, justifying the use of FPT algorithms.

As our main new hardness result, we construct a half-graph-free graph class on which \textsc{Independent Set} is W[1]-hard (\Cref{thm:hardness-is-halfgraph-free}). 

\subsection*{Contribution 2: Approximation}
\textsc{Dominating Set} admits an $\Oof(\log n)$-approximation via the greedy algorithm~\cite{Johnson74a},
while \textsc{Independent Set} and \textsc{Clique} are notoriously hard to approximate:
PCP-based results rule out $n^{1-\varepsilon}$-approximations under standard assumptions~\cite{Hastad96}.
To obtain even better approximation guarantees for \textsc{Dominating Set}, one may consider the VC-dimension of the graph's neighborhood set system.
On graphs of bounded VC-dimension, $\varepsilon$-net methods yield an $\Oof(\log k)$-approximation \cite{BronnimannG95}.
But bounded VC-dimension does not yield comparable improvements for \textsc{Independent Set} or \textsc{Clique}:
for instance, $d$-dimensional box graphs have VC-dimension \(\Oof(d^2)\)~\cite{goldberg93bounding}, but no $\Oof(\log^{d-1} n)$-approximation is known for \textsc{Independent Set}~\cite{boxgraphs}.

We show that bounding the pattern indices \emph{does} yield non-trivial approximations: Given a graph~$G$, we show how to compute \ldots
\begin{enumerate}
    \item \ldots in time \(\Oof(n^{2h+2})\) an independent set of size $k^{1/h}$, \hfill (\Cref{thm:approx-halfgraph})
    \item \ldots in time \(\Oof(n^3\log(n))\) an independent set of size $\log_{2c+2}(k)-2$, \hfill (\Cref{cor:is-comatching})
    % \item On graphs with half-graph index $h$, we find in time \(\Oof(n^{2h+1})\) independent sets of size $k^{1/h}$ (\Cref{thm:approx-halfgraph}).
    % \item On graphs with co-matching index $c$, we find in time \(\Oof(n^3\log(n))\) independent sets of size $\log_{2c+2}(k)-2$ (\Cref{cor:is-comatching}).
\end{enumerate}
where $h$ is the half-graph index of $G$, $c$ is the co-matching index of $G$, and $k$ is the size of the largest independent set in $G$.
The first result is based on a novel branching algorithm with branching depth bounded by the half-graph index.
The second result is an algorithmic utilization of an argument by Gyárfás to show \(\chi\)-boundedness of graph classes excluding an induced path~\cite{A1987}.

\subsection*{Related work}

We highlight three connections of our work. See \Cref{fig:hierarchy} for a diagram of relevant notions.

% Our results on the parameterized complexity of \textsc{Independent Set}, \textsc{Clique}, and \textsc{Dominating Set} can be understood as a continuation of two lines of work:
% as an extension of \emph{sparsity theory} beyond the realm of biclique-free graphs,
% and as a continuation of \emph{stability theory} from model theory.

\paragraph{Connections to sparsity theory.}
The \emph{biclique index} of a graph is the maximum $t$ such that $K_{t,t}$ appears semi-induced.
Accordingly, a class is \emph{biclique-free} if its biclique index is bounded.
Biclique-free graph classes are extremely general sparse graph classes, as they generalize planar graphs, graphs of bounded degeneracy, minor-free graphs, and nowhere dense graph classes.
However, from the perspective of our three target problems, the biclique-free regime is largely under control.
\textsc{Independent Set} and \textsc{Clique} are FPT and polytime solvable, respectively, for trivial reasons:
Any $K_r$ contains $K_{\lfloor r/2\rfloor, \lfloor r/2\rfloor}$ semi-induced, so the clique number is bounded by a class-dependent constant.
By Ramsey-type arguments, bounded clique number implies the existence of unbounded independent sets, and therefore \textsc{Independent Set} is~FPT.
Also \textsc{Dominating Set} has been proven to be~FPT on biclique-free classes~\cite{TelleV12}.

To go beyond sparse classes, the community considered classes that are simultaneously half-graph-free and co-matching-free (also called the \emph{semi-ladder-free classes}).
Semi-ladder-freeness generalizes biclique-freeness, but also includes dense classes such as map graphs~\cite{FabianskiPST19} or more generally all graph classes that can be encoded by existential positive first-order formulas in nowhere dense classes~\cite{MahlmannSiebertz26E}.
Fabiański et al.~\cite{FabianskiPST19} and Guillemot \cite{guillemot2025parameterized}
showed that \textsc{Dominating Set} is FPT even on all semi-ladder-free classes.
Guillemot~\cite{guillemot2025parameterized} asked whether tractability can be extended to classes that are just half-graph-free (and not necessarily co-matching-free).
It is natural to ask the same question for classes that are just co-matching-free (and not necessarily half-graph-free).
Our work answers both questions negatively.

The progressive exploration approach of Fabiański et al.~\cite{FabianskiPST19} also applies to  \textsc{Independent Set}, albeit for this problem, they could only achieve tractability in a strict subfragment of semi-ladder-freeness. Our FPT algorithm for semi-ladder-free classes~(\Cref{thm:IS-kernel}) generalizes their results.

\paragraph{Connections to width measures.}
Eiben et al.~\cite{EibenGHJK22} introduced a unifying framework for width measures based on the notion of $\mathcal{F}$-branchwidth.
In this framework, the same patterns we studied (half-graphs, co-matchings, and matchings) appear as forbidden semi-induced subgraphs between the sides of a cut in branch decompositions. Building on this framework, in~\cite{bergougnoux_et_al:LIPIcs.ESA.2025.16}, the width measure \emph{cham-width} (``chain-anti-matching-width'') was defined, which measures the co-matching and half-graph indices across cuts. 
Classes of bounded cham-width strictly generalize semi-ladder-free classes. 
It was asked in \cite{bergougnoux_et_al:LIPIcs.ESA.2025.16} whether \textsc{Independent Set} is FPT on classes of bounded cham-width.
Our FPT algorithm for semi-ladder-free classes gives hope for a positive answer.

\paragraph{Connections to model theory.}
The half-graph pattern has strong ties to model theory, where its absence characterizes \emph{stable} theories.
\emph{Monadically stable} classes are those in which all relations definable in first-order logic with colors are half-graph-free.
The first-order model checking algorithm for monadically stable classes implies fixed-parameter tractability of \textsc{Independent Set}, \textsc{Clique}, and \textsc{Dominating Set}~\cite{dreier2024first,dreier2023first} on these classes.
The requirement of half-graph-freeness for \emph{all first-order definable relations} excludes many natural graph classes.
It is therefore natural to ask whether tractability can be extended to classes where
\emph{only the plain adjacency relation} is half-graph-free.
Our hardness results (\Cref{thm:hardness-is-halfgraph-free,thm:hardness-ds-halfgraph-free}) show that edge-level half-graph-freeness alone does not suffice for tractability of the three considered problems, justifying the necessity of the stronger assumptions used in stability theory.

\begin{figure}%[htbp]
    \centering
    \includegraphics[scale = 0.8]{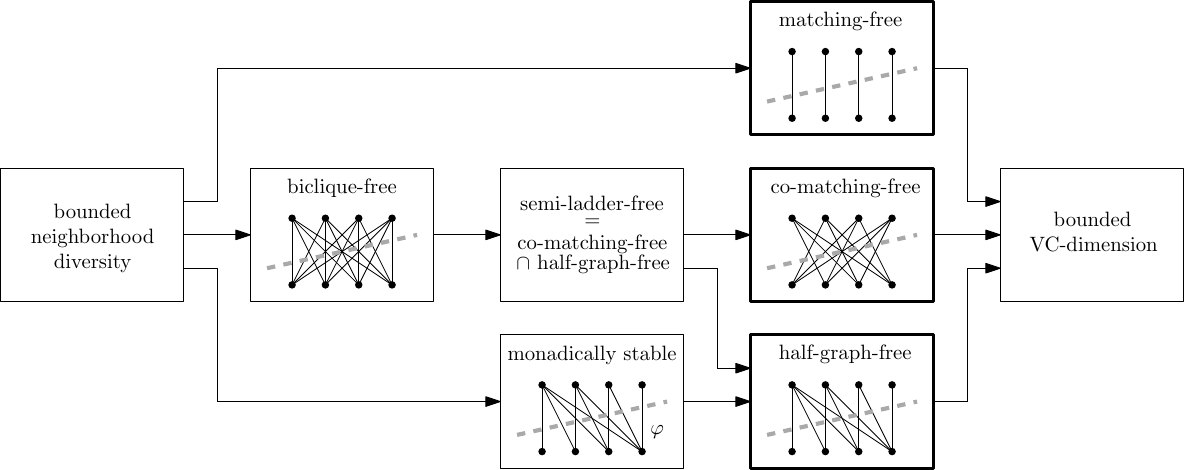}
    \caption{A hierarchy of graph class properties.
    An arrow $P_1 \to P_2$ means every graph class with property $P_1$ also has property $P_2$.
    }
    \label{fig:hierarchy}
\end{figure}

%Recently, a rich theory of \emph{monadic stability} and (more generally) \emph{monadic dependence/NIP} has emerged, where one requires ladder-freeness not only of the edge relation itself but of all definable relations in unary
%expansions of the class, see~\cite{baldwin1985second,braunfeld2025decomposition,dreier2024first,DreierMS23,dreier2024flip,GajarskyMMOPPSS23,dreier_et_al23}.
%This line culminates in the FPT model-checking algorithm for first-order logic on monadically
%stable graph classes~\cite{dreier2024first}, which in particular implies the fixed-parameter tractability of \textsc{Dominating Set, Independent Set} and \textsc{Clique} in monadically stable classes.
%The co-matching-free fragment of monadic stability has
%very recently been investigated in~\cite{MahlmannSiebertz26E}.
%We stress that the theory of monadically stable classes relies on substantially heavier machinery and,
%crucially, asks for ladder exclusion to hold under transductions.
%In contrast, our results in this paper focus on the plain adjacency relation and analyze
%how bounding the ladder/co-matching/matching indices already
%governs the complexity landscape of \textsc{Dominating Set}, \textsc{Independent Set}, and
%\textsc{Clique}.

\section{Preliminaries}\label{sec:prelims}

We use standard notation from graph theory and model theory
and refer to~\cite{Diestel} and~\cite{Hodges} for extensive background.
We write $[m]$ for the set of integers $\{1,\ldots,m\}$.
For a vertex $v$ in a graph, we write $N(v)$ for its open neighborhood (\emph{not} containing $v$) and $N[v] := N(v) \cup \{v\}$ its closed neighborhood.
\section{Hardness of Independent Set for half-graph-free classes}
In this section, we prove the existence of a graph class that has bounded half-graph index, but on which the \textsc{Independent Set} problem is W[1]-hard.
We will do this by providing a reduction from the \textsc{Grid Tiling} problem that does not create large half-graphs.
Our reduction is inspired by a reduction from~\cite{independent_set_in_h_free_graphs} between the same two problems.
We note that the \textsc{Independent Set} instances created in~\cite{independent_set_in_h_free_graphs} have unbounded half-graph index but bounded co-matching index (see also \Cref{sec:hardness-cmi} for a different reduction for the co-matching case).
Conversely, the instance that we create will have bounded half-graph index but unbounded co-matching index. As we later show in \Cref{sec:is-semiladder-FPT}, this tradeoff is unavoidable under standard complexity assumptions, since \textsc{Independent Set} is FPT on every class that has simultaneously bounded half-graph and co-matching index.

\medskip
The \textsc{Grid Tiling} problem was introduced in~\cite{marx2007optimality}.
Its input consists of integers $k, n \in \N$ and a collection~$\mathcal{S}$ containing a non-empty set of \emph{tiles} $S_\pi \subseteq [n] \times [n]$ for each \emph{grid cell} $\pi \in [k]^2$.
The goal is to decide whether there exists
a selection of tiles $s(\pi) \in S_\pi$ for all $\pi \in [k]^2$ such that:
\begin{itemize}
    \item The tiles $s(\pi)$ and $s(\pi_\blacktriangleright)$ have the same first component
    
    for all cells $\pi = (i,j) \in [k-1] \times [k]$ and $\pi_\blacktriangleright := (i+1,j)$; and

    \item The tiles $s(\pi)$ and $s(\pi_\triangledown)$ have the same second component
    
    for all cells $\pi = (i,j) \in [k] \times [k-1]$ and $\pi_\triangledown := (i,j+1)$.
\end{itemize}
See \Cref{tab:gt-example} for an example instance.
It is known that \textsc{Grid Tiling} is W[1]-hard parameterized by the grid size $k$~\cite[Thm.\ 14.28]{cygan2015parameterized}.

\begin{theorem}\label{thm:hardness-is-halfgraph-free}
    There is a graph class with half-graph index at most $256$ on which \textsc{Independent Set} is W[1]-hard.
\end{theorem}
\begin{proof}
Given a \textsc{Grid Tiling} instance $((\mathcal{S},k,n),k)$ we reduce it to
the \textsc{Independent Set} instance $(G_\mathcal{S}, 4k^2)$ by building $G_\mathcal{S}$ as follows. (See \Cref{fig:gt} for an illustration of the construction.)

For every grid cell $\pi \in [k]^2$, we construct the graph $H_\pi$ by first creating four disjoint cliques 
\begin{multicols}{2}
\begin{enumerate}
    \item $U_\pi := \{ u(\pi,\tau) : \tau \in S_\pi\}$,
    \item $R_\pi := \{ r(\pi,\tau) : \tau \in S_\pi\}$,
    \item $D_\pi := \{ d(\pi,\tau) : \tau \in S_\pi\}$,
    \item $L_\pi := \{ \ell(\pi,\tau) : \tau \in S_\pi\}$,
\end{enumerate}
\end{multicols}
named after the four directions \emph{up}, \emph{right}, \emph{down}, and \emph{left}.
Additionally, for every pair of distinct tiles $\tau \neq \tau' \in S_\pi$ we add to $H_\pi$ the edges 
\begin{multicols}{2}
    \begin{itemize}
        \item between $u(\pi,\tau)$ and $r(\pi,\tau')$,
        \item between $r(\pi,\tau)$ and $d(\pi,\tau')$,
        \item between $d(\pi,\tau)$ and $\ell(\pi,\tau')$,
        \item between $\ell(\pi,\tau)$ and $u(\pi,\tau')$.
    \end{itemize}
\end{multicols}
This means $U_\pi$ and $R_\pi$ semi-induce a co-matching, and the same holds for $(R_\pi,D_\pi)$, $(D_\pi,L_\pi)$, and $(L_\pi,U_\pi)$.
We construct $G_\SS$ by taking the disjoint union of all the $H_\pi$ for $\pi \in [k]^2$, and additionally adding edges between
\begin{itemize}
    \item $r(\pi,\tau)$ and $\ell(\pi_\blacktriangleright,\tau')$ if and only if $\tau$ and $\tau'$ have a different first component

    for all $\pi \in [k-1] \times [k]$, $\tau \in S_\pi$, $\tau' \in S_{\pi_\blacktriangleright}$; and
    \item $d(\pi,\tau)$ and $u(\pi_\triangledown,\tau')$ if and only if $\tau$ and $\tau'$ have a different second component,

    for all $\pi \in [k] \times [k-1]$, $\tau \in S_\pi$, $\tau' \in S_{\pi_\triangledown}$.
\end{itemize}
This concludes the construction of $G_\SS$.

\begin{figure}[H]
\centering

\renewcommand{\arraystretch}{1.3}
\begin{tabular}{c|c|c|c}
 & $j=1$ & $j=2$ & $j=3$ \\
\hline
$i=1$ &
$\{\mathbf{(4,4)}, (5,4)\}$ &
$\{\mathbf{(4,5)}, (4,4), (5,5)\}$ &
$\{\mathbf{(4,6)}, (5,4)\}$ \\
\hline
$i=2$ &
$\{\mathbf{(5,4)}, (6,4), (5,5)\}$ &
$\{\mathbf{(5,5)}, (4,5), (5,4), (6,5)\}$ &
$\{\mathbf{(5,6)}, (4,6), (6,6)\}$ \\
\hline
$i=3$ &
$\{\mathbf{(6,4)}, (6,5)\}$ &
$\{\mathbf{(6,5)}, (5,5), (6,6)\}$ &
$\{\mathbf{(6,6)}, (5,5)\}$ \\
\end{tabular}
\captionof{table}{A positive instance of \textsc{Grid Tiling} with $k = 3$ and $n=6$.
The table displays the cells $\pi=(i,j)$ and lists the tiles contained in $S_\pi \subseteq [6]^2$.
A feasible solution is marked in bold.}
\label{tab:gt-example}

\vspace{1em}

\includegraphics[width = \textwidth]{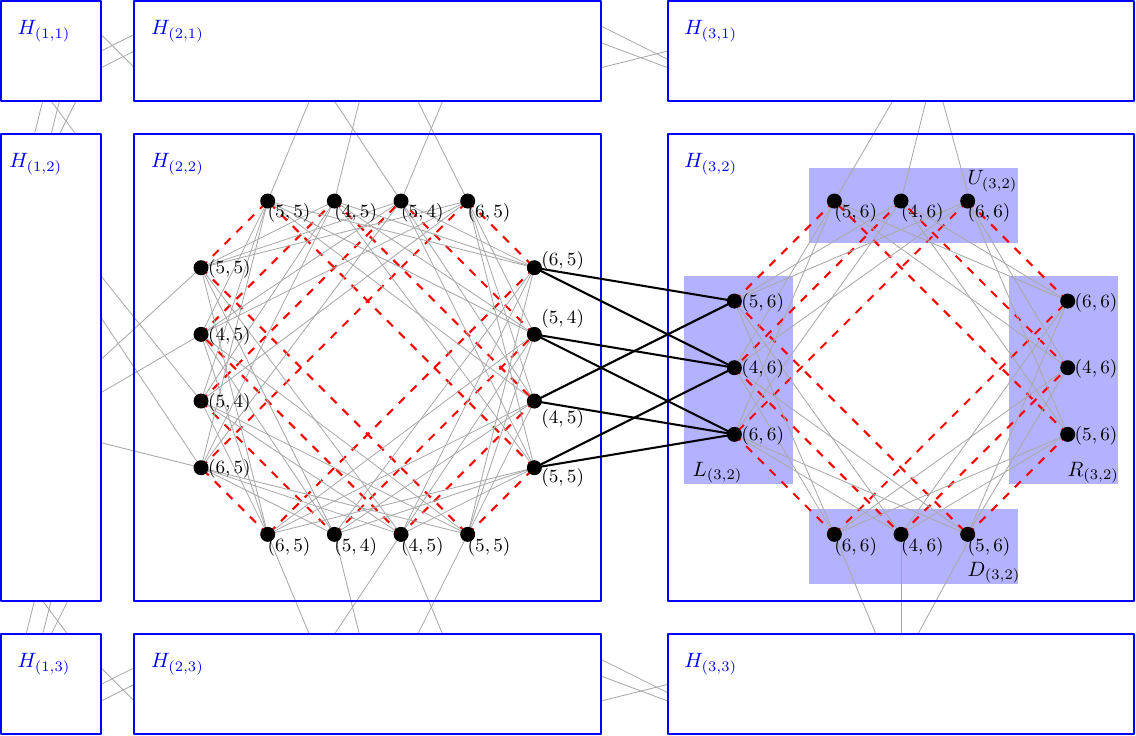}
\captionof{figure}{A depiction of the graph $G_\SS$ for the \textsc{Grid Tiling} instance $\SS$ from \Cref{tab:gt-example}.
Only the edges inside and between the subgraphs $H_{2,2}$ and $H_{3,2}$ are drawn faithfully. The vertices from the other subgraphs have been omitted for the sake of clarity. The red dashed edges are non-edges that we highlighted to illustrate the co-matching connections inside each $H_\pi$.}
\label{fig:gt}
\end{figure}

\begin{claim}
    $(\mathcal{S},k,n)$ is a positive \textsc{Grid Tiling} instance if and only if 
    $(G_\mathcal{S}, 4k^2)$ is a positive \textsc{Independent Set} instance.
\end{claim}
\begin{claimproof}
    Assume there exists a solution for $(\mathcal{S},k,n)$ using tiles $s(\pi) \in S_\pi$  for $\pi \in [k]^2$.
    We claim that
    
    \[
        I = \bigcup_{\pi \in [k]^2} I_\pi
        \quad
        \text{with}
        \quad
        I_\pi := \{
            u(\pi,s(\pi)),
            r(\pi,s(\pi)),
            d(\pi,s(\pi)),
            \ell(\pi,s(\pi))
        \}
    \]
    is an independent set of size $4k^2$.
    Clearly, $I$ has the desired size.
    For each cell $\pi$, we have that $I_\pi$ is an independent set, as all contained vertices belong to the same tile $s(\pi)$.
    By construction, two vertices in distinct sets $x \in I_\pi$ and $y \in I_{\pi'}$ can only be adjacent if, up to symmetry, either
    \begin{itemize}
        \item $\pi' = \pi_\blacktriangleright$, 
        $x = r(\pi,s(\pi))$, 
        $y = \ell(\pi_\blacktriangleright, s(\pi_\blacktriangleright))$, and

        $s(\pi)$ and 
        $s(\pi_\blacktriangleright)$ have a different first component; or

        \item $\pi' = \pi_\triangledown$, 
        $x = d(\pi,s(\pi))$, 
        $y = u(\pi_\triangledown, s(\pi_\triangledown))$, and

        $s(\pi)$ and 
        $s(\pi_\triangledown)$ have a different second component.
    \end{itemize}
    Both possibilities are ruled out due to the $s(\pi)$ forming a solution to \textsc{Grid Tiling}. Hence, $I$ is indeed an independent set.

    Conversely, assume there exists an independent set $I$ of size $4k^2$.
    Since we constructed the $k^2$ graphs $H_\pi$ from four cliques each, 
    $I$ must contain exactly four vertices from each $H_\pi$.
    Since those four vertices are pairwise non-adjacent, 
    they must be of the form
    \[
    \left\{
            u(\pi,s(\pi)),
            r(\pi,s(\pi)),
            d(\pi,s(\pi)),
            \ell(\pi,s(\pi))
    \right\}
    \]
    for some common $s(\pi) \in S_\pi$.
    We now argue that the tiles $s(\pi)$ form a \textsc{Grid Tiling} solution for $(\mathcal{S},k,n)$.
    For every $\pi \in [k-1] \times [k]$, we have that $r(\pi,s(\pi))$ and $\ell(\pi_\blacktriangleright,s(\pi_\blacktriangleright))$ are part of the same independent set $I$ and hence non-adjacent.
    By construction, we must have that $s(\pi)$ and $s(\pi_\blacktriangleright)$ agree on the first component.
    Analogously, for every $\pi \in [k] \times [k-1]$, the tiles $s(\pi)$ and $s(\pi_\triangledown)$ agree on the second component.
    This concludes the proof of the reduction.
\end{claimproof}

By the above claim, we conclude that \textsc{Independent Set} is W[1]-hard on the graph class 
\[
    \CC := \{G_\SS : \SS \text{ is a \textsc{Grid Tiling} instance}\}.
\]
It only remains to bound the half-graph index.

\begin{claim}
    The largest semi-induced half-graph in $G_\mathcal{S}$ has order at most $256$.
\end{claim}

\begin{claimproof}
    Pick $A_0$ and $B_0$ to be sets of maximal size $m = |A_0| = |B_0|$ that semi-induce a half-graph in~$G_\SS$.
    Consider the partition 
    \[
        \PP = \bigcup_{\pi \in [k]^2} \{U_\pi, R_\pi, D_\pi, L_\pi\}    
    \]
    of $V(G_\SS)$ into cliques.
    Applying the pigeonhole principle twice, we can pass to subsets $A \subseteq A_0$ and $B \subseteq B_0$ of size $|A| = |B| \geq \sqrt{\sqrt{m}}$, such that $A$ and $B$ still semi-induce a half-graph, and either 
    \begin{enumerate}
        \item all vertices of $A$ are contained in a single part of $\PP$, or
        \item all vertices of $A$ are contained in pairwise different parts of $\PP$,
    \end{enumerate}
    and one of these two cases also holds for $B$.

    First, suppose that the second case applies to at least one of the two sets. By symmetry, we may assume it applies to $B$.
    Since $A$ and $B$ semi-induce a half-graph, there is a vertex $a \in A$ that is adjacent to all of $B$.
    Observe that in $G_\SS$, each individual vertex has neighbors in at most $4$ different parts of $\PP$. For example, a vertex in the part $L_{(3,2)} \in \PP$ can only have neighbors in $L_{(3,2)},U_{(3,2)},D_{(3,2)}, R_{(2,2)}$ (see \Cref{fig:gt}).
    This bounds $|A| = |B| \leq 4$.

    Otherwise, the first case applies to both $A$ and $B$: there are parts $X,Y \in \PP$ with $A \subseteq X$ and $B \subseteq Y$.
    If $X$ and $Y$ are part of the same subgraph $H_\pi$, then since there are edges between $X$ and~$Y$, they either semi-induce a co-matching or form a clique if $X = Y$. In both cases, each vertex in $A$ can have at most $1$ non-neighbor in $B$. Then the semi-induced half-graph between $A$ and $B$ can have size at most $|A| = |B| \leq 2$.
    
    If $X$ and $Y$ are part of different subgraphs, then up to swapping the roles of $X$ and~$Y$, there exists $\pi \in [k]^2$ such that either
    \begin{center}
        ($X = R_\pi$ and $Y = L_{\pi_\blacktriangleright}$)
        \quad
        or
        \quad
        ($X = D_\pi$ and $Y = U_{\pi_\triangledown}$).
    \end{center}
    We assume the first case, as the second is analogous.
    We claim that $A$ has size at most $2$. Assume, toward a contradiction, that $|A| \geq 3$. Then, since $A$ and $B$ semi-induce a half-graph, there are vertices 
    \[
    r(\pi,\alpha_1), r(\pi,\alpha_2), r(\pi,\alpha_3) \in A \subseteq R_\pi
    \quad\text{and}\quad
    \ell(\pi_\blacktriangleright,\beta_1), \ell(\pi_\blacktriangleright,\beta_2), \ell(\pi_\blacktriangleright,\beta_3) \in B \subseteq L_{\pi_\blacktriangleright},
    \]
    such that $r(\pi,\alpha_i)$ and $\ell(\pi_\blacktriangleright,\beta_j)$ are adjacent if and only if $i\leq j$ for all $i,j \in [3]$.
    This means the tile $\alpha_3$ must have the same first component as the tiles $\beta_1$ and $\beta_2$. However, the tile $\alpha_2$ must have the same first component as $\beta_1$, but a different first component than $\beta_2$.
    This contradicts the fact that ``having the same first component'' is an equivalence relation. Hence, $|A| = |B| \leq 2$.

    Having exhaustively bounded $|A| = |B| \leq 4$, we conclude that $m \leq (|A|^2)^2\leq 4^4 = 256$.
\end{claimproof}
    \noindent
    This concludes the proof of the theorem.
\end{proof}

%\begin{theorem}\label{thm:hardness-is-comatching-free}
%    There is a graph class with bounded co-matching index on which \textsc{Independent Set} is W[1]-hard.
%\end{theorem}
%
%\nikoin{The above hardness already follows from \cite{independent_set_in_h_free_graphs}, which presents a hard graph class, that excludes both an induces $C_4$ and an induced $K_{1,4}$. One can show that there is $t \in \N$ such that every graph that contains a semi-induced co-matching of size $t$ also contains either an induced $C_4$ or $K_{1,4}$. Hence, the class presented in \cite{independent_set_in_h_free_graphs} is co-matching-free.}
%
%\begin{proof}
%    Sketch. We do the same reduction, where the cycle co-matchings and the tile-connecting co-matchings are replaced by ladders.
%    The latter means that we connect $r_{i,j,s}$ and $l_{i+1,j,s'}$ if $s<s'$ and analogously for the columns.
%    Finally, we connect the last row with the first row by a ladder and the last column with the first column by a ladder, analogously.
%    Now observe that in a cycle of ladders we need to pick vertices with the same index. Say we choose index $i_1$ for the first tile. Hence, we need to choose $i_2\leq i_1$ for the second tile, $i_3\leq i_2$ for the third, and so on, and when we close the cycle we have to choose $i_k\leq i_1$, which implies that all indices must be equal.
%\end{proof}

\section{Approximation of Independent Set for half-graph-free classes}

\begin{theorem}\label{thm:approx-halfgraph}
    There is an algorithm that, given a graph $G$, finds in time $O(n^{2h+2})$ an independent set of size at least $k^{1/h}$,
    where $h$ is the half-graph index of $G$ and $k$ is the size of a maximum independent set of $G$.
\end{theorem}
\begin{proof}
We proceed by induction on the half-graph index $h$ of $G$.
If $h = 0$, then $G$ has no edges, so~$V(G)$ itself is an independent set and we return it.
Otherwise, the algorithm branches over all \(u,v \in V(G)\) with \(uv \in E(G)\), recursing into the graphs
\[
    G_{uv} := G[N(v) \setminus N[u]],
\]
and recursively computing independent sets \(D_{uv}\) for all \(uv \in E(G)\).
By construction, \(D_{uv} \cup \{u\}\) is also an independent set.
Finally, the algorithm greedily computes an inclusion-maximal independent set~\(D\)
and returns the largest set among \(D\) and \(D_{uv} \cup \{u\}\) for \(uv \in E(G)\). 

\begin{figure}[h]
    \center
    \includegraphics{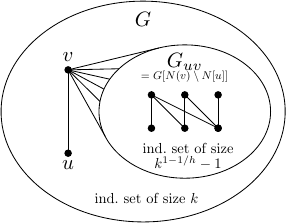}
    \caption{We guess a vertex \(v\) whose neighborhood contains a large part of an independent set and an adjacent vertex \(u\) that is part of the independent set.
    We recurse into the graph \(G_{uv} = G[N(v) \setminus N[u]]\) whose half-graph index has decreased.
    }
    \label{fig:algo}
\end{figure}
As visualized in \Cref{fig:algo},
any semi-induced half-graph in \(G_{uv}\) can be extended by \(u\) and \(v\) to form a semi-induced half-graph of increased order in \(G\).
Hence, \(G_{uv}\) has half-graph index at most \(h-1\),
bounding the recursion depth by \(h\).
At each node, we branch over $O(n^2)$ choices of \(u,v\), and perform $O(n^2)$ work to construct $D$.
Thus, the total running time is $O(n^{2h+2})$.

To argue correctness, assume \(G\) has an independent set \(I\) of size \(k\).
If $|D| \ge k^{1/h}$, then $D$ already achieves the desired bound. 
Note that \(D\) was chosen inclusion-maximal and therefore is a dominating set in \(G\), so the closed neighborhoods \(N[v]\) for \(v \in D\) cover \(I\).
If $|D| < k^{1/h}$, then by the pigeonhole principle there exists $v \in D$ with
\(|N[v] \cap I| \ge |I|/|D| > k^{1-1/h}\).

Since \(|I|/|D| > 1\), the vertex \(v\) has a neighbor in \(I\) and hence \(v \not\in I\).
Thus, since \(v \not\in I\), we have \(|N(v) \cap I| = |N[v] \cap I|\).
There exists \(u \in N(v) \cap I\).
Clearly, \(G_{uv}\) has an independent set \((N(v) \cap I) \setminus \{u\}\) of size at least \(k^{1-1/h} - 1\).
By induction, the recursive call on $G_{uv}$ with half-graph index at most \(h-1\) returns an independent set \(D_{uv}\) of size at least
\[
    \bigl(k^{1-1/h} - 1\bigr)^{1/(h-1)}.
\]
Using the inequality $(x - 1)^a \geq x^a - 1$ for $x \geq 1$ and $0 < a \leq 1$, this is at least
\[
    \bigl(k^{1-1/h}\bigr)^{1/(h-1)} - 1 = k^{1/h} - 1.
\]
After adding back \(u\), the set \(D_{uv} \cup \{u\}\) has the desired lower bound \(k^{1/h}\),
proving the statement.
\end{proof}

\section{FPT Independent Set for semi-ladder-free classes}
\label{sec:is-semiladder-FPT}

In this section, we prove that \textsc{Independent Set} is fixed-parameter tractable on graph classes that have both bounded half-graph index and bounded co-matching index. 
Our FPT algorithm is based on a reduction rule, computable in FPT time, that, when applied exhaustively, yields an equivalent bounded-size instance, which we then solve by brute-force. 
The reduction rule is based on the extraction of \emph{indiscernible sequences}, which interact in a highly regular way with the rest of the graph.
As our key lemma, we prove the following.

\begin{restatable}{lemma}{Homset}\label{cor:homset-neighborhood-structure-k}

There exist (computable) functions $f,g\colon \N^2 \to\N$ such that for every $t\in \N$ the following holds.

\smallskip\noindent
Let $G$ be an $n$-vertex graph with half-graph index $<t$ and co-matching index $<t$.
For every $L\in \N$, if $n\ge f(t,L)$, then~$G$ contains a set $S\subseteq V(G)$ of size $|S|=L$ such that
\begin{enumerate}
\item $S$ is either a clique or an independent set, and
\item for every vertex $w\in V(G)\setminus S$ we have
\[
|N(w)\cap S|<2t
\qquad\text{or}\qquad
S\subseteq N(w).
\]
\end{enumerate}
Moreover, such a set $S$ can be computed in time
$g(t,L)\cdot n$.
\end{restatable}

We defer the proof of \Cref{cor:homset-neighborhood-structure-k} to the next subsection and continue with the presentation of the algorithm.
Crucially, \Cref{cor:homset-neighborhood-structure-k} will allow us to apply the following kernelization rule. 

\begin{lemma}\label{lem:delete-from-clique}
Let $G$ be a graph, let $S\subseteq V(G)$ be a clique, and let $q\in\mathbb N$.
Assume that for every $v\in V(G)\setminus S$
\begin{equation}\label{eq:dagger-few-or-all}
|N(v)\cap S|\le q
\qquad\text{or}\qquad
S\subseteq N(v).
\tag{$*$}
\end{equation}
Let $k\in\mathbb N$ and assume $|S|\ge (k-1)q+2$.
Then for every $s\in S$ the instances $(G,k)$ and $(G-s,k)$ are equivalent for \textsc{Independent Set}.
\end{lemma}

\begin{proof}
Fix $s\in S$.
Any independent set of size $k$ in $G-s$ is also an independent set of size $k$ in $G$.

Conversely, let $I$ be an independent set of size $k$ in $G$.
If $s\notin I$, then $I$ is an independent set of size $k$ in $G-s$ and we are done.
Hence, assume $s\in I$. 
Since $S$ is a clique, we have $I\cap S=\{s\}$.
Let $J:=I\setminus\{s\}$. Then $|J|=k-1$ and $J\cap S=\emptyset$.

For each $x\in J$ we have $xs\not\in E(G)$, hence $x$ is not complete to $S$.
By \eqref{eq:dagger-few-or-all}, it follows that $|N(x)\cap S|\le q$.
Define
\[
F := \bigcup_{x\in J}\bigl(N(x)\cap S\bigr).
\]

Then $|F|\le (k-1)q$.
Since $|S|\ge (k-1)q+2$, there exists $s'\in S\setminus F$ with $s'\neq s$.
For every $x\in J$ we have $s'\notin N(x)\cap S$, hence $xs'\not\in E(G)$.
Therefore $J\cup\{s'\}$ is an independent set of size $k$ in~$G-s$.
\end{proof}

It is now easy to derive an FPT algorithm. 

\begin{theorem}\label{thm:IS-kernel}
There exists a computable function $H(t,k)$ such that for every $t\in \N$ the following holds. 
Given an $n$-vertex graph $G$ with half-graph index $<t$ and co-matching index $<t$, and an integer $k$,
one can decide whether $G$ contains an independent set of size at least $k$ in time $H(t,k)\cdot n^2$.
\end{theorem}

\begin{proof}
Let $(G,k)$ be an instance. 
Set $q:=2t-1$, let $f$ be the function from \cref{cor:homset-neighborhood-structure-k}, and define
\[
K \ :=\ (k-1)q+2,
\qquad
N_0 \ :=\ f(t,K).
\]

We apply our reduction exhaustively:
If $|V(G)|\ge N_0$, compute (using \cref{cor:homset-neighborhood-structure-k} with parameter $K$) a set $S\subseteq V(G)$ of size $|S|= K$ satisfying items~(1) and (2) of the lemma.
\begin{itemize}
\item If $S$ is an independent set, then answer yes and stop.
\item Otherwise, $S$ is a clique. Delete an arbitrary vertex $s\in S$ from $G$ and iterate. 
By \cref{lem:delete-from-clique} the instances $(G,k)$ and $(G-s,k)$ are equivalent.
\end{itemize}

Each application of the reduction rule deletes one vertex from $G$.
Hence after at most $n$ iterations the process terminates with a graph (still denoted by $G$) of size $|V(G)|<N_0$. 
At termination we output the current instance~$(G,k)$.
Since $|V(G)|<N_0=f(t,K)$ and $K=(k-1)(2t-1)+2$, this is an equivalent instance whose size depends only on $t$ and $k$.

In each iteration, we invoke \cref{cor:homset-neighborhood-structure-k} with parameter $K$ on the current graph
on $n_i$ vertices. This takes time $\Oof\bigl(g(t,K)\cdot n_i\bigr)$, where $g$ is the second function from the lemma.
All other work in the iteration is linear in $n_i$.
Summing over all iterations, and using $n_i\le n$ for all $i$, we obtain the
running time 
\[
\sum_i \Oof\bigl(g(t,K)\cdot n_i\bigr)\ =\ \Oof\bigl(g(t,K)\cdot n^2\bigr).
\]

On the resulting instance we solve the problem by brute-force and define $H(t,k)$ accordingly. 
\end{proof}

In the remainder of this section we will prove \Cref{cor:homset-neighborhood-structure-k}.

\subsection{Indiscernible sequences}

We assume basic familiarity with first-order logic, as can be found, e.g., in the textbooks~\mbox{\cite{Hodges,TentZiegler12}}.
% For a first-order formula $\varphi$, we denote by $\free(\varphi)$ the set of its \emph{free variables}, that is, those variables that are not bound by a quantifier. 
% Accordingly, writing $\varphi(x_1,\ldots,x_\kappa)$ indicates that $\free(\varphi)\subseteq\{x_1,\ldots,x_\kappa\}$.
%The \emph{quantifier rank} $\qr(\varphi)$ is the nesting depth of quantifiers in $\varphi$.
%Finally, the \emph{length} $|\varphi|$ is the number of symbols of $\varphi$ in some fixed reasonable encoding of formulas (hence well-defined up to constant factors), and we use it only as a coarse syntactic size measure.
Indiscernible sequences are a central tool in model theory. Let us recall their definition. 
Let $G$ be a graph, let $\phi(x_1,\ldots,x_\kappa)$ be a formula in first-order logic
and let $I=(a_i)_{i\in[n]}$ be a sequence of vertices of~$G$.
We say that $I$ is \emph{$\phi$-indiscernible} if for all indices
\[
1 \le i_1 < \cdots < i_\kappa \le n
\quad\text{and}\quad
1 \le j_1 < \cdots < j_\kappa \le n
\]
we have
\[
G \models \phi(a_{i_1},\ldots,a_{i_\kappa})
\quad\Longleftrightarrow\quad
G \models \phi(a_{j_1},\ldots,a_{j_\kappa}).
\]
For a finite set $\Delta$ of formulas we say that $I$ is \emph{$\Delta$-indiscernible}
if it is $\phi$-indiscernible for all $\phi\in \Delta$.
In the following, we will sometimes treat an indiscernible sequence as a set, e.g., when considering the induced subgraph $G[I]$ or when intersecting with other sets. 

\begin{example}\label{obs:clique-or-independent}
Let $G$ be a graph and let $\eta(x,y):=E(x,y)$.
If $I=(v_1,\dots,v_n)$ is an $\eta$-indiscernible sequence in $G$, then the induced subgraph $G[I]$ is either a clique or an independent set.
\end{example}

By Ramsey's Theorem~\cite{ramsey}, every sufficiently long sequence of vertices contains a long indiscernible subsequence,
as formalized by the following lemma.
(See \cite[Thm.\ 4.18]{jukna2011extremal} for a modern presentation.)

\begin{lemma}\label{lem:indiscernibles-exist} 
There exists a (computable) function $f$ such that for every finite set of formulas $\Delta$ and every sequence of elements of length at least $f(k,m)$, the sequence contains a $\Delta$-indiscernible subsequence of length~$m$, where
\[
    k := |\Delta| \cdot \text{maximum number of free variables appearing in a formula of $\Delta$}.
\]
\end{lemma}

% The computability of the function follows by the fact that the proof in~\cite{EhrenfeuchtMostowski56} is constructive (see also the algorithm of~\cite[Theorem~15]{DBLP:journals/talg/KreutzerRS19}). 
% In our case we can actually settle with a simple brute-force method to extract sequences with our desired properties, as we will present below. 
% All we need is that the bound~$f(m)$ is computable. 

\subsection{Neighborhood dichotomies from indiscernibles}

It was observed in multiple places that, when working in well-behaved graph classes and choosing the right formula sets, indiscernible sequences can impose a lot of structure on their neighborhoods~\cite{blumensath2011,KreutzerRS19,MalliarisShelah2014RegularityStableGraphs, dreier_et_al23,dreier2022combinatorial,dreier2024flip}.
We will prove a lemma in this spirit for graph classes that are half-graph-free and co-matching-free (\Cref{lem:Chi-complete-uniform}).
A formula set that almost satisfies our needs was used by Malliaris and Shelah~\cite{MalliarisShelah2014RegularityStableGraphs} to establish a strong regularity lemma for well-behaved graph classes. 
We will work with a formula set that is slightly different from that set, since their definition of half-graph does not require that the vertices $a_1,\ldots, a_t, b_1,\ldots, b_t$ forming the half-graph are pairwise different. 
We have to incorporate this assumption into our formulas. 
Let 
\begin{align*}
\chi_{2t}(x_1,\dots,x_{2t})
\ :=\
\exists y\Bigl(\ \bigwedge_{1\le i\le t}E(x_i,y)\ \wedge\ \bigwedge_{t+1\le i\le 2t}\neg E(x_i,y)\ \wedge \bigwedge_{1\leq i\leq 2t} y\neq x_i\Bigr),
\\
\chi_{2t}^\star(x_1,\dots,x_{2t})
\ :=\
\exists y\Bigl(\ \bigwedge_{1\le i\le t}\neg E(x_i,y)\ \wedge\ \bigwedge_{t+1\le i\le 2t}E(x_i,y)\ \wedge \bigwedge_{1\leq i\leq 2t} y\neq x_i\Bigr).
\end{align*}

The formula $\chi_{2t}$ states that there exists a vertex $y$ different from all $x_i$ that is adjacent to exactly the first $t$ elements $x_1,\ldots, x_t$ and not to the last $t$ elements $x_{t+1},\ldots, x_{2t}$. 
The formula~$\chi_{2t}^\star$ symmetrically states that $y$ is adjacent to the last $t$ and not adjacent to the first $t$ elements $x_i$.
See \Cref{fig:formulas} for an illustration.

\begin{figure}[htbp]
    \centering
    \includegraphics[width = \textwidth]{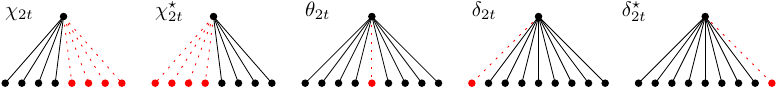}
    \caption{The connection patterns of the vertex $y$ (on the top) towards a tuple from the indiscernible sequence (on the bottom) that is verified by the formulas $\chi_{2t}, \chi_{2t}^\star, \theta_{2t},\delta_{2t}, \delta^\star_{2t}$ where $t = 4$.}
    \label{fig:formulas}
\end{figure}

We obtain the following variation of~\mbox{\cite[Claim 3.2]{MalliarisShelah2014RegularityStableGraphs}}.

\begin{lemma}\label{lem:Gamma-dichotomy}
Let $G$ be a graph with half-graph index $<t$. 
Let $I=(v_1,\dots,v_n)$ be an $\{\chi_{2t},\chi_{2t}^\star\}$-indiscernible sequence in $G$ of length $n\ge 4t$.
Then for every vertex $w\in V(G)\setminus I$,
\[
|N(w)\cap I|<2t
\qquad\text{or}\qquad
|I\setminus N(w)|<2t.
\]
\end{lemma}

\begin{proof}
Fix $w\in V(G)\setminus I$ and assume for contradiction that
\begin{equation}\label{eq:Gamma-both-large-new}
|N(w)\cap I|\ge 2t
\qquad\text{and}\qquad
|I\setminus N(w)|\ge 2t.
\end{equation}
Let $A:=(v_1,\dots,v_{2t})$ be the prefix of $I$ of length $2t$.
By~\eqref{eq:Gamma-both-large-new}, at least one of the two sets
$A\cap N(w)$ and $A\setminus N(w)$ has size at least $t$.
We assume without loss of generality that
\begin{equation}\label{eq:Gamma-wlog-neigh}
|A\cap N(w)|\ge t,
\end{equation}
otherwise we argue symmetrically using $\chi_{2t}^\star$ instead of $\chi_{2t}$.

Since the second inequality in~\eqref{eq:Gamma-both-large-new} guarantees at least $2t$ non-neighbors of $w$ in the whole sequence,
and $A$ has length $2t$, we can choose indices
$1\le i_1<\cdots<i_t\le 2t$ and $2t< j_1<\cdots< j_t\le n$
such that $v_{i_1},\dots,v_{i_t}\in N(w)$ and $v_{j_1},\dots,v_{j_t}\notin N(w)$.
Hence
\[
G\models \chi_{2t}(v_{i_1},\dots,v_{i_t},v_{j_1},\dots,v_{j_t})
\]
with witness $y=w$. 
Note that $w\notin I$, so the conjunct $\bigwedge_{1\le i\le 2t}y\neq x_i$ is satisfied as well.

Since $I$ is $\chi_{2t}$-indiscernible, it follows that every increasing $2t$-tuple from $I$ satisfies $\chi_{2t}$.
In particular, for every $i\in[t]$ we have
\begin{equation}\label{eq:Gamma-sliding}
G\models \chi_{2t}(v_i,v_{i+1},\dots,v_{i+2t-1}).
\end{equation}
Fix, for each $i\in[t]$, a witness $y_i\in V(G)$ for~\eqref{eq:Gamma-sliding}.
Then
\[
y_i v_{i+j}\in E(G)\ \ \text{for}\ \ 0\le j\le t-1,
\qquad\text{and}\qquad
y_i v_{i+j}\notin E(G)\ \ \text{for}\ \ t\le j\le 2t-1,
\]
and moreover $y_i\neq v_{i+j}$ for all $0\le j\le 2t-1$ by the conjunct $\bigwedge_{1\le i\le 2t}y\neq x_i$.

For $j\in[t]$ define $r_j:=v_{t+j-1}$, so $(r_1,\dots,r_t)=(v_t,\dots,v_{2t-1})$.
As $y_i$ is adjacent precisely to the first $t$ vertices of the window $(v_i,\dots,v_{i+2t-1})$, i.e.\ to
$v_i,\dots,v_{i+t-1}$, and $r_j=v_{t+j-1}$, we have $r_jy_i\in E(G)$ if and only if
\[
r_j\in\{v_i,\dots,v_{i+t-1}\}
\ \ \Leftrightarrow\ \ 
t+j-1 \le i+t-1
\ \ \Leftrightarrow\ \ 
j\le i.
\]

Observe that the $y_i$ are distinct from the $r_j$ by construction, pairwise different since they have different neighborhoods, and the $v_i$ are pairwise distinct as they are different vertices of $I$. 
Hence, $(r_1,\dots,r_t)$ and $(y_1,\dots,y_t)$ witness a half-graph of order~$t$ (up to swapping the sides),
contradicting the assumption that the half-graph index of $G$ is $<t$.
This contradiction shows that~\eqref{eq:Gamma-both-large-new} cannot hold, and the lemma follows.
\end{proof}

We add the assumption of bounded co-matching index and consider indiscernibility with respect to further formulas. Let
\begin{align*}
\theta_{2t}(x_1,\dots,x_{2t+1}) 
\ :=\ &
\exists y\Bigl(\ \bigwedge_{1\le i\le t}E(x_i,y)\ \wedge\ \neg E(x_{t+1},y)\ \wedge\ \bigwedge_{t+2\le i\le 2t+1}E(x_i,y)\ \wedge \bigwedge_{1\leq i\leq 2t+1} y\neq x_i\Bigr),\\
\delta_{2t}(x_1,\ldots, x_{2t+1}) \ :=\ & 
\exists y\Bigl(\ \neg E(x_1,y) \wedge \bigwedge_{2\le i\le 2t+1}E(x_i,y)\wedge \bigwedge_{1\le i\le 2t+1}y\neq x_i\Bigr), \\
\delta_{2t}^\star(x_1,\ldots, x_{2t+1}) \ :=\ & 
\exists y\Bigl(\ \neg E(x_{2t+1},y) \wedge \bigwedge_{1\le i\le 2t}E(x_i,y)\wedge \bigwedge_{1\le i\le 2t+1}y\neq x_i\Bigr), 
\end{align*}

The formula $\theta_{2t}$ states that there exists an element $y$ that is different from all $x_i$ and adjacent to all $x_i$ except for $x_{t+1}$ in the middle. 
The formula $\delta_{2t}$ states that there exists an element $y$ different from all~$x_i$ that is non-adjacent to~$x_1$ and adjacent to $x_i$ for $i\geq 2$, and symmetrically, $\delta_{2t}^\star$ states that there exists~$y$ different from all~$x_i$ that is non-adjacent exactly to the last element $x_{2t+1}$. In particular, when $\delta_{2t}$ or~$\delta_{2t}^\star$ are satisfied, then~$y$ has a non-neighbor and $2t$ neighbors among $x_1,\ldots, x_{2t+1}$. See \Cref{fig:formulas} for an illustration.

\bigskip

The following lemma is a generalization of an observation made for the more restricted monadically stable, co-matching-free graph classes from~\cite{dreier2022combinatorial} to the more general setting of half-graph-free, co-matching-free classes.
Let \[\Gamma_t:=\{\eta, \chi_{2t}, \chi_{2t}^\star, \theta_{2t},\delta_{2t}, \delta^\star_{2t}\}.\]

\begin{restatable}{lemma}{ChiCompleteUniform}\label{lem:Chi-complete-uniform}
Let $t\in\mathbb N$ and let $G$ be a graph of half-graph index $<t$ and co-matching index $<t$.
Let $I=(v_1,\dots,v_n)$ be a $\Gamma_t$-indiscernible sequence in $G$ of length $n\ge 8t+1$.
Then for every vertex $w\in V(G)\setminus I$ we have
\[
|N(w)\cap I|<2t
\qquad\text{or}\qquad
I\subseteq N(w).
\]
\end{restatable}

\setcounter{equation}{0}
\begin{proof}
Fix $w\in V(G)\setminus I$.
By \cref{lem:Gamma-dichotomy} we have
$|N(w)\cap I|<2t$ or $|I\setminus N(w)|<2t$.
In the first case we are done, so assume
\begin{equation}\label{eq:ccu-few-nonneigh}
|I\setminus N(w)|<2t.
\end{equation}

If $I\subseteq N(w)$, we are done as well; hence assume for contradiction that there exists
\begin{equation}\label{eq:ccu-pick-u}
u\in I\setminus N(w).
\end{equation}

Assume without loss of generality that $u$ occurs among the first $4t+1$ vertices of the sequence,
i.e.\ $u=v_p$ with $p\le 4t+1$; otherwise we argue symmetrically with $\delta_{2t}^\star$.
Since $n\ge 8t+1$, there are at least $4t$ vertices of $I$ after $u$.
By~\eqref{eq:ccu-few-nonneigh}, among these vertices there are at most $2t-1$ non-neighbors of~$w$,
hence in the suffix after $u$ there are at least $2t$ neighbors of $w$.
Choose indices
\[
p< \alpha_1<\cdots<\alpha_{2t}\le n
\qquad\text{with}\qquad
v_{\alpha_1},\dots,v_{\alpha_{2t}}\in N(w).
\]
Then the increasing $(2t{+}1)$-tuple
\[
(u,\,v_{\alpha_1},\dots,v_{\alpha_{2t}})
\]
satisfies $\delta_{2t}$ with witness $y=w$ (where $w\notin I$ ensures the conjunct $y\neq x_i$ is satisfied).

Since $I$ is $\delta_{2t}$-indiscernible, it follows that $\delta_{2t}$ holds on every increasing $(2t{+}1)$-tuple from $I$.
In particular, we have
\begin{equation}\label{eq:ccu-delta-shift}
G\models \delta_{2t}(v_{4t+1},v_{4t+2},\dots,v_{6t+1}).
\end{equation}
Fix a witness $w'\in V(G)$ for~\eqref{eq:ccu-delta-shift} (which satisfies $w'\neq v_i$ for all $i\in\{4t+1,\dots,6t+1\}$ as the conjunct $\bigwedge_{1\le i\le 2t+1}y\neq x_i$ of $\delta_{2t}$ is satisfied).
Let $u'=v_{4t+1}$, to which $w'$ is not adjacent as $\delta_{2t}$ is satisfied. 

As $w'$ is not adjacent to all vertices of $I$ and has at least $2t$ neighbors in $I$, again by \cref{lem:Gamma-dichotomy}, $|I\setminus N(w')|<2t$.
As \(u\) is sufficiently centered in the sequence, we can choose $t$ neighbors of $w'$ that occur before $u'$ and $t$ neighbors that occur after~$u'$. 
Thus we obtain an increasing $(2t{+}1)$-tuple
\[
(c_1',\dots,c_t',\,u',\,c_{t+2}',\dots,c'_{2t+1})
\]
from $I$ such that $w'c_i'\in E(G)$ for all $i\in[2t+1]\setminus \{t+1\}$, while $w'u'\not\in E(G)$. 
Hence, $G\models \theta_{2t}$ for this tuple. 
Since $I$ is $\theta_{2t}$-indiscernible, it follows that for every $i\in[t]$ we have
\begin{equation}\label{eq:theta-sliding-polished}
G\models \theta_{2t}(v_i,v_{i+1},\dots,v_{i+2t}).
\end{equation}
Fix, for each $i\in[t]$, a witness $y_i\in V(G)$ for~\eqref{eq:theta-sliding-polished}.
Then
\[
y_iv_{i+j}\in E(G)\ \text{ for } j\in\{0,\dots,t-1\}\cup\{t+1,\dots,2t\},
\qquad\text{and}\qquad
y_iv_{i+t}\notin E(G),
\]
and moreover $y_i\neq v_{i+j}$ for all $0\le j\le 2t$ by the last conjunct in $\theta_{2t}$.

Define
\[
c_j := v_{t+j}\qquad (j=1,\dots,t+1),
\]
so $(c_1,\dots,c_{t+1})=(v_{t+1},v_{t+2},\dots,v_{2t+1})$.
By the same sliding-window argument as before, for all \mbox{$i,j\in[t]$} we have
\begin{equation}\label{eq:ccu-comatching-pattern}
E(y_i,c_j)\quad\Longleftrightarrow\quad i\neq j.
\end{equation}

It also follows from the construction that all involved vertices are pairwise distinct.
Hence the bipartite pattern between $\{c_1,\dots,c_t\}$ and $\{y_1,\dots,y_t\}$ semi-induces a co-matching of order $t$,
contradicting the assumption that the co-matching index of $G$ is $<t$.
This contradiction shows that our assumption~\eqref{eq:ccu-pick-u} was false, and therefore $I\subseteq N(w)$.
\end{proof}

We can now derive \Cref{cor:homset-neighborhood-structure-k}, which we restate for convenience.

\Homset*

\begin{proof}
Let $h\colon\N\times\N\to\N$ be the function from \cref{lem:indiscernibles-exist} and let $k=|\Gamma_t|\cdot (2t+1)$. 
Then, by the lemma, every vertex sequence of length at least $h(k,m)$ contains a $\Gamma_t$-indiscernible subsequence of length~$m$.
Let $L_0 := \max(L,8t+1)$ and define $f(t,L) := h(k,L_0)$.

Assume $n\ge f(t,L)$.
Take an arbitrary ordering $(v_1,\ldots,v_n)$ of $V(G)$ and consider its prefix $P$ of length~$f(t,L)$. 
We can pick such a prefix because $f$ is computable. 
By definition of $f(t,L)$, $P$ contains an increasing $\Gamma_t$-indiscernible subsequence
$S=(u_1,\ldots,u_{L_0})$ of length $L_0$.

\begin{enumerate}
    \item Since $\eta\in\Gamma_t$, the sequence $S$ is $\eta$-indiscernible.
By \cref{obs:clique-or-independent}, $G[S]$ is either a clique or an independent set.
    \item \cref{lem:Chi-complete-uniform} applies to $S$ and yields
    $|N(w)\cap S|<2t$ or $S\subseteq N(w)$.
\end{enumerate}

If $L_0>L$, take any subset $S'\subseteq S$ of size $L$. 
The induced subgraph on $S'$ is still a clique or independent set, and for every $w\in V(G)\setminus S'$ we have
$|N(w)\cap S'|<2t$ or $S'\subseteq N(w)$.

For the running time we can just apply a brute-force algorithm. 
We iterate over all subsets $S$ of $P$ of size $L$ and verify whether $S$ satisfies the two conditions. 
Generating all such $S$ takes time $\Oof(2^{f(t,L)})$. 
Checking if $S$ induces a clique or independent set can be done in time $\Oof(L^2\cdot n)$, as an adjacency test between any two vertices of $S$ takes time $\Oof(n)$, assuming an adjacency list encoding. 
Also, testing if every vertex $w\in V(G)\setminus S$ has either fewer than $2t$ neighbors in $S$ or is adjacent to all of $S$ can be done in time $\Oof(L\cdot n)$. 
For this, equip every vertex $w\in V(G)$ with a counter. 
We iterate over all vertices $s\in S$ and increase the counter of every neighbor of $s$ by one when considering it. 
When done, we iterate once over all vertices of $V(G)\setminus S$ and check whether their counter is at most $2t$ or equal to~$|S|$. 
Defining $g$ accordingly yields the claimed running time. 
\end{proof}

\section{Hardness of Dominating Set for half-graph-free classes}

In this section, we prove the existence of a graph class that has bounded half-graph index, but on which the \textsc{Dominating Set} problem is W[1]-hard.

\begin{theorem}\label{thm:hardness-ds-halfgraph-free}
    There is a graph class with half-graph index at most 16 on which \textsc{Dominating Set} is W[1]-hard.
\end{theorem}
\begin{proof}
We follow the standard parameterized reduction from \textsc{Multicolored Independent Set} 
to \textsc{Dominating Set} of Cygan et al.~(see \cite[Thm.\ 13.9]{cygan2015parameterized}).
Given an instance $(G, k, (V_1,\ldots,V_k))$ of \textsc{Multicolored Independent Set}, the reduction constructs a graph $G'$ as follows:
\begin{enumerate}
    \item For every vertex $v \in V(G)$, introduce $v$ in $G'$.
    \item For every $1 \leq i \leq k$, make $V_i$ a clique in $G'$.
    \item For every $1 \leq i \leq k$, introduce two new vertices $x_i, y_i$ into $G'$ and
    make them adjacent to every vertex of $V_i$. Denote \(V_i^* = V_i \cup \{x_i,y_i\}\).
    \item For every edge $e \in E(G)$ with endpoints $u \in V_i$ and $v \in V_j$, introduce
    a vertex $w_e$ into $G'$ and make it adjacent to every vertex of $(V_i \cup V_j) \setminus \{u, v\}$.
    Denote \(W = \{ w_e \mid e \in E(G)\}\).
\end{enumerate}
Cygan et al.~\cite{cygan2015parameterized} argue the correctness of this reduction. 
It remains to argue that $G'$ has half-graph index at most 16.
Suppose, for contradiction, that there is a semi-induced half-graph consisting of \(a_1,\dots,a_{17}\) and \(b_1,\dots,b_{17}\).
Since \(W\) is an independent set and each pair \(a_i,b_i\) is adjacent, at most one of \(a_i,b_i\) can lie in \(W\). Hence at most 17 half-graph vertices lie in \(W\).
By symmetry, assume the \(a\)-side contains at most 8 \(W\)-vertices.
Thus at least 9 \(a\)-vertices lie in \(V_1^*,\dots,V_k^*\).
If some \(V_i^*\) contained 5 \(a\)-vertices, then some \(b\)-vertex would be
adjacent to one vertex from \(V_i^*\) and non-adjacent to the other four.
This cannot happen, as every vertex has either no neighbors or at most 3 non-neighbors in \(V_i^*\).
Hence, each set \(V_i^*\) can contain at most 4 \(a\)-vertices,
and therefore the 9 \(a\)-vertices must be distributed over at least three different \(V_i^*\)-sets.
But then \(b_{17}\), which is adjacent to all \(a\)-vertices, must be adjacent to vertices from at least three different \(V_i^*\)-sets.
This is a contradiction, as no vertex in \(G'\) has neighbors in more than two \(V_i^*\)-sets.
\end{proof}

\section{Hardness results for co-matching-free classes}
\label{sec:hardness-cmi}

In this section, we derive hardness results for co-matching-free classes from the literature.
A graph is a \emph{unit square graph} if it is an intersection graph of axis-parallel unit squares.
The following two hardness results were proven by Marx.

\begin{theorem}[{\cite[Theorem 1]{marx06}}]
    \textsc{Dominating Set} is W[1]-hard on unit square graphs.
\end{theorem}

\begin{theorem}[{\cite[Theorem 1]{marx05}}]
    \textsc{Independent Set} is W[1]-hard on unit square graphs.
\end{theorem}

Let $\overline{3K_2}$ be the complement of the disjoint union of three $K_2$, i.e., the complement of the matching of order $3$. Neuen observed the following obstructions for unit square graphs.

\begin{lemma}[{\cite[Lemma 3.3]{neuen16}}]\label{lem:forbidden_squares}
    Every unit square graph excludes $K_{1,5}$ and $\overline{3K_2}$ as induced subgraphs.
\end{lemma}

\begin{corollary}\label{cor:unitquare_comatching}
    The class of unit square graphs has bounded co-matching index.
\end{corollary}
\begin{proof}
    Let $R(k)$ denote the Ramsey number, i.e., the smallest number such that every graph on $R(k)$ vertices contains either an independent set or a clique of size $k$.
    Suppose, toward a contradiction, that $G$ has a semi-induced co-matching of size $R(R(6))$ with sides $A$ and $B$.
    By Ramsey, there exists $A' \subseteq A$ of size $R(6)$ that is a clique or an independent set.
    Restricting to the corresponding matched vertices in $B$ and applying Ramsey to that subset, we obtain a semi-induced co-matching of size~$6$ where both sides are homogeneous.

    If one side is an independent set $\{a_1, \ldots, a_6\}$, then any $b_j$ on the other side is adjacent to the five independent vertices $\{a_i : i \neq j\}$, so $G$ contains $K_{1,5}$.
    If both sides are cliques, then any three matched pairs $(a_1, b_1), (a_2, b_2), (a_3, b_3)$ induce $\overline{3K_2}$.
    Both cases contradict \Cref{lem:forbidden_squares}.
\end{proof}

\section{Approximation of Independent Set for co-matching-free classes}

A \emph{proper $k$-coloring} of a graph $G$ is an assignment $\lambda : V(G) \to [k]$ such that every two adjacent vertices $u$ and $v$ in $G$ satisfy $\lambda(u) \neq \lambda(v)$.
The \emph{chromatic number} $\chi(G)$ denotes the smallest integer~$k$ such that there exists a proper $k$-coloring of $G$.
The \emph{clique number} $\omega(G)$ denotes the size of the largest clique in $G$.
A graph class $\CC$ is \emph{$\chi$-bounded} if there exists a function $f :\N \to \N$ that relates the two parameters:
\[
    \chi(G) \leq f(\omega(G)) \quad
    \text{for all $G\in \CC$.}
\]
We refer to \cite{scott2020survey} for a survey on $\chi$-boundedness.
We will derive an approximation algorithm for the \textsc{Independent Set} problem on co-matching-free classes, building on the following classic argument by Gyárfás. Here, $P_t$ denotes the path on $t$ vertices.
\begin{theorem}[Theorem 2.4 of \cite{A1987}]
    $P_t$-free graphs are $\chi$-bounded.
    In particular, every graph \(G\) contains either an induced \(K_k\), an induced \(P_t\), or a proper \(t^k\)-coloring of \(G\).
\end{theorem}
The above statement is constructive, and it is folklore that it can be turned into a polynomial-time algorithm that produces a witness.
As we are not aware of any explicit runtime analysis in the literature, we provide a proof of the following statement in the appendix.

\begin{restatable}[folklore]{theorem}{GyarfasAlgo}\label{thm:Gyarfas}
There exists an \(O(n^3)\)-time algorithm that, given a graph $G$ and integers $k,t \in \N$, returns an induced \(K_k\), an induced \(P_t\), or a proper \(t^k\)-coloring of \(G\).
\end{restatable}

\begin{corollary}\label{cor:clique-matching}
    There is an algorithm that, given a graph $G$, finds in time $O(n^{3}\log(n))$ a clique of size at least $\log_{2m+2}(k)-2$,
    where $m$ is the matching index of $G$ and $k$ is the clique number of $G$.
\end{corollary}
\begin{proof}
    A \(P_{2m+2}\) has matching index \(m+1\), and thus \(G\) contains no induced \(P_{2m+2}\).
    Thus, \cref{thm:Gyarfas} computes for a given \(k'\) and \(t=2m+2\) either an induced \(K_{k'}\) or a proper \((2m+2)^{k'}\)-coloring of~\(G\).
    Since \(\omega(G) = k\), we have \(\chi(G) \geq k\).
    Hence, for \(k'\) such that \((2m+2)^{k'} < k\),
    \cref{thm:Gyarfas} must return a clique~\(K_{k'}\).
    The largest such \(k'\) is \(\lfloor \log_{2m+2}(k-1) \rfloor \geq \log_{2m+2}(k) - 2\).
    Binary search over \(k' \in \{1, \ldots, n\}\) yields the result in time \(O(n^3 \log(n))\).
\end{proof}
Running the above algorithm on the complement graph yields an approximation for \textsc{Independent Set} in classes of bounded co-matching index.

\begin{corollary}\label{cor:is-comatching}
    There is an algorithm that, given a graph $G$, finds in time $O(n^{3}\log(n))$ an independent set of size at least $\log_{2m+2}(k)-2$,
    where $m$ is the co-matching index of $G$ and $k$ is the size of the largest independent set in $G$.
\end{corollary}

\bibliographystyle{plain}
\bibliography{references}

\appendix

\section{Proof of \Cref{cor:nd-trichotomy}}\label{sec:proof-nd-trichotomy}
\NdTrichotomy*
\begin{proof}
    Let $Q_{\mathrm{bi}}$ be the function from \Cref{thm:trichotomy} and $R$ the Ramsey number, and set
    \(Q(h) := R\bigl(Q_{\mathrm{bi}}(h)\bigr)\).
    Let $G$ be a graph with neighborhood diversity at least \(Q(h)\), and let $S\subseteq V(G)$ be a twin-free set of size \(Q(h)\)
    (take one representative from each twin class).
    By Ramsey's theorem, $S$ contains a clique or independent set $A$ of size $Q_{\mathrm{bi}}(h)$.

    Since $A$ is homogeneous, two vertices in $A$ are twins if and only if they have identical neighborhoods in $V(G)\setminus A$.
    As $A$ is twin-free, all vertices in $A$ have pairwise distinct neighborhoods in $V(G)\setminus A$.
    Let $B$ contain one vertex for each distinct neighborhood in $A$.
    The bipartite graph between $A$ and $B$ is twin-free on both sides, so \Cref{thm:trichotomy} yields
    an induced matching, co-matching, or half-graph of order $h$, which is semi-induced in $G$.
\end{proof}

\section{Proof of \Cref{thm:Gyarfas}}

\GyarfasAlgo*

\begin{proof}
Define
\[
    f(k,t) = \begin{cases}
        0 & \text{if } k = 1, \\
        (t-1)(f(k-1,t)+1) & \text{otherwise,}
    \end{cases}
    \qquad \text{and}\qquad 
    g(k,t,l) = (l-1)(f(k-1,t)+1).
\]
Note that \(f(k,t) = g(k,t,t)\).
Moreover, \(f(k,t) < t^k\): For \(k=1\), we have \(f(1,t) = 0 < t\).
For \(k \ge 2\), by induction,
\[
    f(k,t) = (t-1)(f(k-1,t)+1) < (t-1)(t^{k-1}+1) = t^k - t^{k-1} + t - 1 < t^k,
\]
where the last inequality uses \(t^{k-1} \ge t > t - 1\).

We describe two mutually recursive algorithms.
\begin{itemize}
    \item
        \(\Gyarfas(G,k,t)\) takes a non-empty graph \(G\), and integers \(k,t \in \mathbb{N}\),
        and returns either an induced \(P_t\) or \(K_k\), or a proper \(f(k,t)\)-coloring of \(G\).
    \item
        \(\GyarfasSub(G,k,t,v,l)\) takes a connected graph \(G\), integers \(k,t,l \in \mathbb{N}\) with \(l \le t\),
        and a start vertex \(v \in V(G)\),
        and returns either an induced \(P_t\) or \(K_k\), or a proper \(g(k,t,l)\)-coloring of \(G\), or an induced \(P_l\) with endpoint \(v\).
\end{itemize}

\paragraph{Description of \(\Gyarfas(G,k,t)\).}
    For each connected component \(G'\) of \(G\) do the following:

    \medskip\noindent
    Pick an arbitrary \(v \in V(G')\) and call \(\GyarfasSub(G',k,t,v,t)\).
    If one subroutine call yields an induced \(P_t\) or \(K_k\), return it.
    Otherwise every component has a proper coloring with \(g(k,t,t) = f(k,t)\) colors.
    Aggregate the colorings of the connected components into a proper \(f(k,t)\)-coloring of \(G\)
    and return it.

\paragraph{Description of \(\GyarfasSub(G,k,t,v,l)\).}
    If \(t=1\) or \(k=1\) or \(l=1\), return the single vertex \(v\) (as \(P_t\), \(K_1\), or \(P_l\) with endpoint $v$, respectively).
    If \(G\) consists of a single vertex \(v\), return the coloring of \(v\) with one color.
    Otherwise, because $G$ is connected, $G[N(v)]$ is a non-empty graph and we can call \(\Gyarfas(G[N(v)],k-1,t)\).
    If the call returns a \(P_t\), return it.
    If it returns a \(K_{k-1}\) in \(G[N(v)]\), extend it by \(v\) and return the resulting \(K_k\).
    Otherwise, we obtain a proper \(f(k-1,t)\)-coloring of \(G[N(v)]\).
    Then, for every connected component \(H\) of \(G[V(G) \setminus N[v]]\) do the following.

    \medskip\noindent
    Pick an arbitrary \(w \in N(v)\) that is adjacent to a vertex in \(H\), and call \(\GyarfasSub(G[V(H) \cup \{w\}],k,t,w,l-1)\).
    If one call yields a \(P_t\) or \(K_k\), return it.
    If one yields a \(P_{l-1}\) with endpoint~\(w\), extend it into a \(P_l\) with endpoint \(v\) and return it.
    This is valid as \(w\) is adjacent to \(v\) and the remainder of the \(P_{l-1}\) lives in a component of \(G[V(G) \setminus N[v]]\)
    and thus is non-adjacent to \(v\).
    Otherwise, all \(H\) have a proper \(g(k,t,l-1)\)-coloring.
    Return a proper coloring of \(G\) by using \(g(k,t,l-1)\) colors for \(V(G)\setminus N[v]\),
    \(f(k-1,t)\) colors for \(N(v)\), and one color for \(v\).
    As required, the total number of colors is bounded by
    \begin{align*}
        g(k,t,l-1) + f(k-1,t) + 1 &=  (l- 1 -1)(f(k-1,t)+1) + f(k-1,t) + 1 \\
        &= (l-1)(f(k-1,t) + 1) \\
        &= g(k,t,l).
    \end{align*}
\paragraph{Runtime.}

Let \(T(n)\) be the maximum runtime of \(\GyarfasSub\) over all inputs with \(n\)-vertex graphs.
We argue by induction on \(n\) that \(T(n) \le cn^3\) for some constant \(c\).
For \(n=1\) this is clear, so assume \(n \ge 2\).

Consider an \(n\)-vertex graph \(G\) and parameters \(k,t,v,l\).
The routine \(\GyarfasSub(G,k,t,v,l)\) performs \(O(n^2)\) work and additionally 
recurses into the connected components of \(G[V(G) \setminus N[v]]\),
and (via a call to the main \(\Gyarfas\)-routine) recurses into the connected components of \(G[N(v)]\).
For each connected component of \(G[V(G) \setminus N[v]]\), an additional vertex is added before recursing.

Let \(n^+_1,\dots,n^+_r\) be the sizes of the connected components of \(G[N(v)]\),
and let \(n^-_1,\dots,n^-_m\) be the sizes of the connected components of \(G[V(G) \setminus N[v]]\).
Hence, by choosing \(c\) large enough, we can upper bound the runtime \(T\) of \(\GyarfasSub(G,k,t,v,l)\) by
\[
   T \le cn^2 + \sum_{i=1}^{r} T(n^+_i) + \sum_{i=1}^{m} T(1 + n^-_i).
\]

Let \(a=n^+_1\) be the first number and \(B = \{n^+_2,\dots,n^+_r,n^-_1,\dots,n^-_m\}\) be the set of remaining numbers.
We can upper bound the above term by 
\[
   T \le cn^2  + T(a) + \sum_{b' \in B} T(1+b').
\]
As \(v\) is in none of the components, all components are strictly smaller than \(n\).
By induction, we can substitute
\[
   T \le cn^2 + ca^3 + \sum_{b' \in B} c(1+b')^3.
\]
The function \(d(x):=(x+1)^3\) is superadditive for integers \(x,y\ge 1\), as
\[
d(x+y)-d(x)-d(y) =(x+y+1)^3-(x+1)^3-(y+1)^3 =3x^2y+3xy^2+6xy-1 \ge 0.
\]
Superadditivity yields
\[
   \sum_{b' \in B} (1+b')^3 \le \Bigl(1+\sum_{b' \in B} b'\Bigr)^3.
\]
Let \(b = 1+\sum_{b' \in B}b'\).
As \(v\) appears in no connected component, \(a+b = n\).
Thus
\[
   T \le c(a+b)^2 + ca^3 + cb^3.
\]
For integers \(a,b \ge 1\), we have \(a + b \le 2ab \le 3ab\), hence \((a+b)^2 \le 3ab(a+b) = 3a^2b + 3ab^2\), and thus
\[
   T/c \leq a^3 + b^3 + (a+b)^2 \le a^3 + 3a^2b + 3ab^2 + b^3 = (a+b)^3 = n^3.
\]
Therefore, \(\GyarfasSub(G,k,t,v,l)\) runs in time \(T \le c n^3\).
As we made no assumptions on the input except \(|V(G)|=n\), we have \(T(n) \le c n^3\), as claimed.
\end{proof}

\end{document}